\documentclass[runningheads]{llncs}
\usepackage{makeidx}
\usepackage{amsmath}
\usepackage{amssymb}

\begin{document}

\title{Inapproximability of Maximum Weighted Edge Biclique and Its Applications}

\author{Jinsong Tan}

\institute{Department of Computer and Information Science \\
School of Engineering and Applied Science \\
University of Pennsylvania, Philadelphia, PA 19104, USA \\
\email{ jinsong@seas.upenn.edu } }

\maketitle

\begin{abstract} Given a bipartite graph $G = (V_1,V_2,E)$ where edges take
on {\it both} positive and negative weights from set
$\mathcal{S}$, the {\it maximum weighted edge biclique} problem,
or $\mathcal{S}$-MWEB for short, asks to find a bipartite subgraph
whose sum of edge weights is maximized. This problem has various
applications in bioinformatics, machine learning and databases and
its (in)approximability remains open. In this paper, we show that
for a wide range of choices of $\mathcal{S}$, specifically when
$\left| \frac{\min\mathcal{S}} {\max \mathcal{S}} \right| \in
\Omega(\eta^{\delta-1/2}) \cap O(\eta^{1/2-\delta})$ (where $\eta
= \max\{|V_1|, |V_2|\}$, and $\delta \in (0,1/2]$), no polynomial
time algorithm can approximate $\mathcal{S}$-MWEB within a factor
of $n^{\epsilon}$ for some $\epsilon > 0$ unless $\mathsf{RP =
NP}$. This hardness result gives justification of the heuristic
approaches adopted for various applied problems in the
aforementioned areas, and indicates that good approximation
algorithms are unlikely to exist. Specifically, we give two
applications by showing that: 1) finding statistically significant
biclusters in the SAMBA model, proposed in \cite{Tan02} for the
analysis of microarray data, is $n^{\epsilon}$-inapproximable; and
2) no polynomial time algorithm exists for the Minimum Description
Length with Holes problem \cite{Bu05} unless $\mathsf{RP=NP}$.
\end{abstract}

\section{Introduction}

Let $G= (V_1, V_2, E)$ be an undirected bipartite graph. A {\it
biclique subgraph} in $G$ is a complete bipartite subgraph of $G$
and {\it maximum edge biclique} (MEB) is the problem of finding a
biclique subgraph with the most number of edges. MEB is a
well-known problem and received much attention in recent years
because of its wide range of applications in areas including
machine learning \cite{Mis03}, management science \cite{Swa98} and
bioinformatics, where it is found particularly relevant in the
formulation of numerous biclustering problems for biological data
analysis \cite{Che00,BenDor02,Tan02,Zhang02,Tan05}, and we refer
readers to the survey by Madeira and Oliveira \cite{Mad04} for a
fairly extensive discussion on this. Maximum edge biclique is
shown to be NP-hard by Peeters \cite{Pee00} via a reduction from
3SAT. Its approximability status, on the other hand, remains an
open question despite considerable efforts
\cite{Fei02,Fei04,Sub04} \footnote{Note it might be easy to
confuse the MEB problem with the {\it Bipartite Clique} problem
discussed by Khot in \cite{Sub04}. {\it Bipartite Clique}, which
also known as {\it Balanced Complete Bipartite Subgraph}
\cite{Fei04}, aims to maximize the number of vertices of a {\it
balanced} subgraph whereas MEB aims to maximize the total weights
on edges in a (not necessarily balanced) subgraph.}. In
particular, Feige and Kogan \cite{Fei04} conjectured that maximum
edge biclique is hard to approximate within a factor of
$n^\epsilon$ for some $\epsilon > 0$. In this paper, we consider a
weighted formulation of this problem defined as follows

\begin{definition}
\noindent {\bf $\mathcal{S}$-Maximum Weighted Edge Biclique
($\mathcal{S}$-MWEB) }

\noindent {\bf Instance:} A complete bipartite graph $G = (V_1,
V_2, E)$ (throughout the paper, let $\eta = max\{|V_1|, |V_2|\}$
and $n=|V_1|+|V_2|$), a weight function $w_G : E \rightarrow
\mathcal{S}$, where $\mathcal{S}$ is a set consisting of {\it
both} positive and negative integers.

\noindent {\bf Question:} Find a biclique subgraph of $G$ where
the sum of weights on edges is maximized.
\end{definition}

A few comments are in order. First note it is not a lose of
generality but a technical convenience to require the graph be
complete, one can always think of an incomplete bipartite graph as
complete where non-edges are assigned weight 0. Also note we
require that both positive and negative weights be in
$\mathcal{S}$ at the same time because otherwise
$\mathcal{S}$-MWEB becomes a trivial problem.

\noindent Our study of $\mathcal{S}$-MWEB is motivated by the
problem of finding statistically significant biclusters in
microarray data analysis in the SAMBA model \cite{Tan02} and the
Minimum Description Length with Holes (MDLH) problem
\cite{Bu04,Bu05,Guha05}; detailed discussion of the two problems
can be found in Sect. 4. Our main technical contribution of this
paper is to show that if $\mathcal{S}$ satisfies the condition $|
\frac{\min \mathcal{S} } {\max \mathcal{S} } | \in
\Omega(\eta^{\delta-1/2}) \cap O(\eta^{1/2-\delta})$, where
$\delta > 0$ is any arbitrarily small constant, then no polynomial
time algorithm can approximate $\mathcal{S}$-MWEB within a factor
of $n^{\epsilon}$ for some $\epsilon > 0$ unless $\mathsf{RP =
NP}$. This result enables us to answer open questions regarding
the hardness of the SAMBA model and the MDLH problem. Since
maximum edge biclique can be characterized as a special case of
$\mathcal{S}$-MWEB with $\mathcal{S} = \{-\eta, 1\}$, the
$n^{\epsilon}$-inapproximability result also provides interesting
insights into the conjectured $n^{\epsilon}$-inapproximability
\cite{Fei04} of maximum edge biclique.

The rest of the paper is organized in three sections. In Sect. 2,
we present the main technical result by proving the aforementioned
inapproximability of $\mathcal{S}$-MWEB. We give applications of
this by answering hardness questions regarding two applied
problems in Sect. 3. We conclude this work by raising a few open
problems in the last section.

\section {Approximating $\mathcal{S}$-Maximum Edge Biclique is Hard}

We start this section by giving two lemmas about CLIQUE, which
will be used in establishing inapproximability for the biclique
problems we consider later. Lemma \ref{lemma0} is a recent result
by Zuckerman \cite{Zuc06}, obtained by a derandomization of
results of H\r{a}stad \cite{Has99}; Lemma \ref{lemma1} follows
immediately from Lemma \ref{lemma0}.

\begin{lemma}
\label{lemma0} {\bf(\cite{Zuc06})} It is NP-hard to approximate
CLIQUE within a factor of $n^{1-\epsilon}$, for any $\epsilon >
0$.
\end{lemma}

\begin{lemma}
\label{lemma1} For any constant $\epsilon > 0$, no polynomial time
algorithm can approximate CLIQUE within a factor of
$n^{1-\epsilon}$ with probability at least $\frac{1}{poly(n)}$
unless $\mathsf{RP=NP}$.
\end{lemma}

\subsection{A Technical Lemma}

\noindent We first describe the construction of a structure called
$\{\gamma, \{\alpha, \beta\}\}$-Product, which will be used in the
proof of our main technical lemma.

\begin{definition} {\bf$(\{\gamma, \{\alpha, \beta\}\}$-Product)}

\noindent{\bf Input:} An instance of $\mathcal{S}$-MWEB on
complete bipartite graph $G = V_1 \times V_2$, where $\gamma \in
\mathcal{S}$ and $\alpha < \gamma < \beta$; an integer $N$.

\noindent{\bf Output:} Complete bipartite graph $G^N = V_1^N
\times V_2^N$ constructed as follows: $V_1^N$ and $V_2^N$ are $N$
duplicates of $V_1$ and $V_2$, respectively. For each edge
$(i,j)\in G^N$, let $(\phi(i),\phi(j))$ be the corresponding edge
in $G$. If $w_G(\phi(i),\phi(j)) = \gamma$, assign weight $\alpha$
or $\beta$ to $(i,j)$ independently and identically at random with
expectation being $\gamma$, denote the weight by random variable
$X$. If $w_G(\phi(i),\phi(j)) \neq \gamma$, then keep the weight
unchanged. Call the weight function constructed this way
$w(\cdot)$.

For any subgraph $H$ of $G^N$, denote by $w_{\gamma}(H)$ (resp.,
$w_{-\gamma}(H)$) the total weight of $H$ contributed by
former-$\gamma$-edges (resp., other edges). Clearly, $w(H) =
w_{\gamma}(H) + w_{-\gamma}(H)$.
\end{definition}

\noindent With a graph product constructed in this randomized
fashion, we have the following lemma.

\begin{lemma} \label{boostingLemma} Given an $\mathcal{S}$-MWEB
instance $G = (V_1, V_2, E)$ where $\gamma \in \mathcal{S}$, and a
number $\delta \in (0,\frac{1}{2}]$; let $\eta =
\max{(|V_1|,|V_2|)}$, $N =
\eta^{\frac{\delta(3-2\delta)+3}{\delta(1+2\delta)} }$, $G^N =
(V_1^N, V_2^N, E)$ be the $\{\gamma, \{\alpha, \beta\} \}$-product
of $G$ and $\mathcal{S'} = ( \mathcal{S} \cup \{\alpha, \beta\} )
- \{ \gamma \}$. If

1. $|\beta - \alpha| = O((N\eta)^{\frac{1}{2}-\delta})$; and

2. there is a polynomial time algorithm that approximates the
$\mathcal{S'}$-MWEB instance within a factor of $\lambda$, where
$\lambda$ is some arbitrary function in the size of the
$\mathcal{S'}$-MWEB instance

then there exists a polynomial time algorithm that approximates
the $\mathcal{S}$-MWEB instance within a factor of $\lambda$, with
probability at least $\frac{1}{poly(n)}$.
\end{lemma}
\begin{proof} For notational convenience, we denote $\eta^{\frac{1}{2}-\delta}$
by $f(\eta)$ throughout the proof. Define random variable $Y = X
-\gamma$, clearly $E[Y] = 0$. Suppose there is a polynomial time
algorithm $\mathbb{A}$ that approximates $\mathcal{S'}$-MWEB
within a factor of $\lambda$, we can then run $\mathbb{A}$ on
$G^N$, the output biclique $G_B^*$ corresponds to $N^2$ bicliques
in $G$ (not necessarily all distinct). Let $G_A^*$ be the most
weighted among these $N^2$ subgraphs of $G$, in the rest of the
proof we show that with high probability, $G_A^*$ is a
$\lambda$-approximation of $\mathcal{S}$-MWEB on $G$.

Denote by $\mathbb{E}_1$ the event that $G_B^*$ does not imply a
$\lambda$-approximation on $G$. Let $\mathcal{H}$ be the set of
subgraphs of $G^N$ that do not imply a $\lambda$-approximation on
$G$, clearly, $|\mathcal{H}| \leq 4^{N\eta}$. Let $H'$ be an
arbitrary element in $\mathcal{H}$, we have the following
inequalities $$
\begin{array}{lllr}
Pr\left\{ \mathbb{E}_1 \right\} & \leq Pr\left\{ \mbox{at least one element in $\mathcal{H}$ is a $\lambda$-approximation of $G^N$} \right\} \\
& \leq 4^{N\eta} \cdot Pr \left\{ H' \mbox{ is a $\lambda$-approximation of $G^N$} \right\} \\
& = 4^{N\eta} \cdot Pr\{ \mathbb{E}_2 \}
\end{array}
$$ where $\mathbb{E}_2$
is the event that $H'$ is a $\lambda$-approximation of $G^N$.

Let the weight of an optimal solution $U_1 \times U_2$ of $G$ be
$K$, denote by $U_1^N \times U_2^N$ the corresponding
$N^2$-duplication in $G^N$. Let $x_1$ and $x_2$ be the number of
former-$\gamma$-edges in $H'$ and $U_1^N \times U_2^N$,
respectively. Suppose $\mathbb{E}_2$ happens, then we must have
$$
\begin{array}{cc}
w_{-\gamma}(H') + x_1 \gamma \leq N^2(\frac{K}{\lambda}-1) \\
w_{-\gamma}(H') + w_{\gamma}(H')  \geq  \frac{1}{\lambda} ( w_{-\gamma}(U_1^N \times U_2^N) + w_{\gamma}(U_1^N \times U_2^N) ) \\
\end{array}
$$ where the first inequality follows from the fact that we only
consider integer weights. Since $w_{-\gamma}(U_1^N \times U_2^N) =
N^2 K - x_2 \gamma$, it implies $$ ( w_{\gamma}(H') - x_1 \gamma )
- \frac{1}{\lambda} ( w_{\gamma}(U_1^N \times U_2^N) - x_2 \gamma)
\geq N^2$$ so we have the following statement on probability
$$
\begin{array}{lr}
Pr\{ \mathbb{E}_2 \} & \leq  Pr \left\{ ( w_{\gamma}(H') - x_1
\gamma ) - \frac{1}{\lambda} ( w_{\gamma}(U_1^N \times U_2^N) -
x_2 \gamma ) \geq N^2 \right\}
\end{array}
$$ Let $z_1$ (resp., $z_2$ and $z_3$) be the number of edges in
$E(H') - E(U_1^N \times U_2^N)$ (~resp., $E(U_1^N \times U_2^N) -
E(H')$ and $E(U_1^N \times U_2^N) \cap E(H') ~$) transformed from
former-$\gamma$-edges in $G$. We have $$
\begin{array}{llr}
& Pr \left\{ ( w_{\gamma}(H') - x_1 \gamma ) - \frac{1}{\lambda} (
w_{\gamma}(U_1^N \times U_2^N) - x_2 \gamma ) \geq N^2
\right\} \\

= & Pr \left\{ \sum_{i=1}^{z_1}{Y_i} - \frac{1}{\lambda}
\sum_{j=1}^{z_2}{Y_j} + \frac{\lambda-1}{\lambda}
\sum_{k=1}^{z_3}{Y_k}
\geq N^2 \right\} \\

= & Pr \left\{ \sum_{i=1}^{z_1}{Y_i} + \frac{1}{\lambda}
\sum_{j=1}^{z_2}{(-Y_j)} + \frac{\lambda-1}{\lambda}
\sum_{k=1}^{z_3}{Y_k}
\geq N^2 \right\} \\

\leq & Pr \left\{ \sum_{i=1}^{z_1}{Y_i} \geq \frac{N^2}{3}
\right\} + Pr \left\{ \frac{1}{\lambda} \sum_{j=1}^{z_2}{(-Y_j)}
\geq \frac{N^2}{3} \right\} + Pr \left\{ \frac{\lambda-1}{\lambda}
\sum_{k=1}^{z_3}{Y_k} \geq
\frac{N^2}{3} \right\} \\

\leq & Pr \left\{ \sum_{i=1}^{z_1}{Y_i} \geq \frac{N^2}{3}
\right\} + Pr \left\{\sum_{j=1}^{z_2}{(-Y_j)} \geq \frac{N^2}{3}
\right\} + Pr \left\{ \sum_{k=1}^{z_3}{Y_k} \geq
\frac{N^2}{3} \right\} \\

\leq & \sum_{i\in \{1,2,3\}}^{} { \left( \exp \left( -2z_i
\left(\frac{N^2}{3z_i(c_1 f(N\eta))} \right)^2 \right) \right) } & \mbox{\hspace{-9em} (Hoeffding bound)} \\

\leq & 3 \cdot \exp \left( -c_2 \cdot
\frac{N^{1+2\delta}}{\eta^{3-2\delta}}  \right) & \mbox{\hspace{-9em} ($z_i \leq \eta^2 N^2$)}\\
\end{array}
$$ where $c_1, c_2$ are constants ($c_2>0$). Now if we set $N = \eta^{\frac{3-2\delta}{1+2\delta} + \theta }$ for some $\theta$, we have
$$ Pr\left\{\mathbb{E}_1 \right\} \leq 4^{N\eta} \cdot Pr \left\{
\mathbb{E}_2 \right\} \leq 3 \cdot \exp{\left( \ln{4} \cdot
\eta^{\frac{4}{(1+2\delta)} + \theta} -c_2 \cdot \eta^{(1+
2\delta) \theta} \right)} $$

For this probability to be bounded by $\frac{1}{2}$ as $\eta$ is
large enough, we need to have $ \frac{4}{1+2\delta} + \theta < (1
+ 2\delta) \theta$. Solving this inequality gives $\theta
> \frac{2}{\delta (1+2\delta)}$. Therefore, for any $\delta \in (0, \frac{1}{2}]$, by setting $N =
\eta^{\frac{\delta(3-2\delta)+3}{\delta(1+2\delta)}}$, we have
$Pr\{\mathbb{E}_1\}$, i.e. the probability that the solution
returned by $\mathbb{A}$ does not imply a $\lambda$-approximation
of $G$, is bounded from above by $\frac{1}{2}$ once input size is
large enough. This gives a polynomial time algorithm that
approximates $\mathcal{S}$-MWEB within a factor of $\lambda$ with
probability at least $\frac{1}{2}$. \qed
\end{proof}

This lemma immediately leads to the following corollary.

\begin{corollary}\label{boostingCoro} Following the construction in Lemma
\ref{boostingLemma}, if $\mathcal{S'}$-MWEB can be approximated
within a factor of $n^{\epsilon'}$, for some $\epsilon'>0$, then
there exists a polynomial time algorithm that approximates
$\mathcal{S}$-MWEB within a factor of $n^\epsilon$, where
$\epsilon = (1+\frac{\delta(3-2\delta) + 3}{\delta (1+2\delta)})
\epsilon'$, with probability at least $\frac{1}{poly(n)}$.
\footnote{Note we are slightly abusing notation here by always
representing the size of a given problem under discussion by $n$.
Here $n$ refers to the size of $\mathcal{S'}$-MWEB (resp.
$\mathcal{S}$-MWEB) when we are talking about approximation factor
$n^{\epsilon'}$ (resp. $n^{\epsilon}$). We adopt the same
convention in the sequel.}
\end{corollary}
\begin{proof}
Let $|G|$ and $|G^N|$ be the number of nodes in the
$\mathcal{S}$-MWEB and $\mathcal{S'}$-MWEB problem, respectively.
Since $\lambda = |G^N|^{\epsilon'} \leq |G|^{(1+
\frac{\delta(3-2\delta)+3}{\delta(1+2\delta)})\epsilon'}$, our
claim follows from Lemma \ref{boostingLemma}. \qed
\end{proof}

\subsection{$\{-1,0,1\}$-MWEB}

In this section, we prove inapproximability of $\{-1,0,1\}$-MWEB
by giving a reduction from CLIQUE; in subsequence sections, we
prove inapproximability results for more general
$\mathcal{S}$-MWEB by constructing randomized reduction from
$\{-1,0,1\}$-MWEB.

\begin{lemma}
\label{MWEBhard2} The decision version of the $\{-1,0,1\}$-MWEB
problem is $\mathsf{NP}$-complete.
\end{lemma}
\begin{proof} We prove this by describing a reduction from CLIQUE. Given a
CLIQUE instance $G = (V,E)$, construct $G' = (V',E')$ such that
$V' =  V_1 \cup V_2 $ where $V_1$, $V_2$ are duplicates of $V$ in
that there exist bijections $\phi_1 : V_1 \rightarrow V$ and
$\phi_2 : V_2 \rightarrow V$. And $$
\begin{array}{lllr}
E' & = & E_1 \cup E_2 \cup E_3 \\
E_1 & = & \{(u,v) ~|~ u \in V_1, v \in V_2 \mbox{ and } (\phi_1(u), \phi_2(v))\in E\} \\
E_2 & = & \{(u,v) ~|~ u \in V_1, v \in V_2, \phi_1(u) \neq \phi_2(v) \mbox{ and } (\phi_1(u),\phi_2(v))\notin E\} \\
E_3 & = & \{(u,v) ~|~ u \in V_1, v \in V_2, \mbox{ and } \phi_1(u)
= \phi_2(v)\}
\end{array}
$$

Clearly, $G'$ is a biclique. Now assign weight 0 to edges in
$E_1$, $-1$ to edges in $E_2$ and 1 to edges in $E_3$. We then
claim that there is a clique of size $k$ in $G$ if and only if
there is a biclique of total edge weight $k$ in $G'$.

First consider the case where there is a clique of size $k$ in
$G$, let $U$ be the set of vertices of the clique, then taking the
subgraph induced by $\phi_1^{-1}(U) \times \phi_2^{-1}(U)$ in $G'$
gives us a biclique of total weight $k$.

Now suppose that there is a biclique $U_1 \times U_2$ of total
weight $k$ in $G'$. Without loss of generality, assume $U_1$ and
$U_2$ correspond to the same subset of vertices in $V$ because if
$(\phi_1(U_1)-\phi_2(U_2)) \cup (\phi_2(U_2)-\phi_1(U_1))$ is not
empty, then removing $(U_1-U_2) \cup (U_2-U_1)$ will never
decrease the total weight of the solution. Given $\phi_1(U_1) =
\phi_2(U_2)$, we argue that there is no edge of weight $-1$ in
biclique $U_1 \times U_2$; suppose otherwise there exists a weight
$-1$ edge $(i_1,j_2)$ ($i_1\in U_1$, and $j_2\in U_2$), then the
corresponding edge $(j_1,i_2)$ ($j_1\in U_1$, and $i_2\in U_2$)
must be of weight $-1$ too and removing $i_1, i_2$ from the
solution biclique will increase total weight by at least 1 because
among all edges incident to $i_1$ and $i_2$, $(i_1,i_2)$ is of
weight 1, $(i_1,j_2)$ and $(i_2,j_1)$ are of weight $-1$ and the
rest are of weights either 0 or $-1$.

Therefore, we have shown that if there is a solution $U_1 \times
U_2$ of weight $k$ in $G'$, $U_1$ and $U_2$ correspond to the same
set of vertices $U \in V$ and $U$ is a clique of size $k$. It is
clear that the reduction can be performed in polynomial time and
the problem is $\mathsf{NP}$, and thus $\mathsf{NP}$-complete.
\qed
\end{proof}

Given Lemma \ref{lemma0}, the following corollary follows
immediately from the above reduction.

\begin{theorem}
\label{coro1} For any constant $\epsilon > 0$, no polynomial time
algorithm can approximate problem $\{-1,0,1\}$-MWEB within a
factor of $n^{1-\epsilon}$ unless $\mathsf{P = NP}$.
\end{theorem}
\begin{proof} It is obvious that the reduction given in the proof of
Lemma \ref{MWEBhard2} preserves inapproximability exactly, and
given that CLIQUE is hard to approximate within a factor of
$n^{1-\epsilon}$ unless $\mathsf{P = NP}$, the theorem follows.
\qed
\end{proof}

\begin{theorem}
\label{coro2} For any constant $\epsilon > 0$, no polynomial time
algorithm can approximate $\{-1,0,1\}$-MWEB within a factor of
$n^{1-\epsilon}$ with probability at least $\frac{1}{poly(n)}$
unless $\mathsf{RP=NP}$.
\end{theorem}
\begin{proof} If there exists such a randomized algorithm for $\{-1,0,1\}$-MWEB,
combining it with the reduction given in Lemma \ref{MWEBhard2}, we
obtain an $\mathsf{RP}$ algorithm for CLIQUE. This is impossible
unless $\mathsf{RP = NP}$. \qed
\end{proof}

\subsection{$\{-1,1\}$-MWEB}

\begin{lemma} \label{mainThm1} If there exists a polynomial time
algorithm that approximates $\{-1,1\}$-MWEB within a factor of
$n^{\epsilon}$, then there exists a polynomial time algorithm that
approximates $\{-1,0,1\}$-MWEB within a factor of $n^{5\epsilon}$
with probability at least $\frac{1}{poly(n)}$.
\end{lemma}

\begin{proof} We prove this by constructing a $\{\gamma, \{\alpha,
\beta\}\}$-Product from $\{-1,0,1\}$-MWEB to $\{-1,1\}$-MWEB by
setting $\gamma = 0$, $\alpha = -1$ and $\beta = 1$. Since $\delta
= \frac{1}{2}$, according to Corollary \ref{boostingCoro}, it is
sufficient to set $N = \eta^4$ so that the probability of
obtaining a $n^{5\epsilon}$-approximation for $\{-1,0,1\}$-MWEB is
at least $\frac{1}{poly(n)}$. \qed
\end{proof}

\begin{theorem} \label{inapprox1} For any constant $\epsilon >0$,
no polynomial time algorithm can approximate $\{-1,1\}$-MWEB
within a factor of $n^{\frac{1}{5}-\epsilon}$ with probability at
least $\frac{1}{poly(n)}$ unless $\mathsf{RP = NP}$.
\end{theorem}
\begin{proof} This follows directly from Theorem \ref{coro2} and Lemma
\ref{mainThm1}.\qed
\end{proof}

\subsection{$\{-\eta^{\frac{1}{2}-\delta}, 1\}$-MWEB and $\{-\eta^{\delta -\frac{1}{2}}, 1\}$-MWEB}

In this section, we consider the generalized cases of the
$\mathcal{S}$-MWEB problem.

\begin{theorem}
\label{BigCoro} For any $\delta \in (0,\frac{1}{2}]$, there exists
some constant $\epsilon$ such that no polynomial time algorithm
can approximate $\{-\eta^{\frac{1}{2}-\delta},1\}$-MWEB within a
factor of $n^\epsilon$ with probability at least
$\frac{1}{poly(n)}$ unless $\mathsf{RP = NP}$. The same statement
holds for $\{-\eta^{\delta-\frac{1}{2}},1\}$-MWEB.
\end{theorem}
\begin{proof} We prove this by first construct a $\{\gamma, \{\alpha,
\beta\}\}$-Product from $\{-1,1\}$-MWEB to
$\{-\eta^{\frac{1}{2}-\delta},1\}$-MWEB by setting $\gamma = -1$,
$\alpha = -(N \eta)^{\frac{1}{2}-\delta}$ and $\beta = 1$. By
Corollary \ref{boostingCoro}, we know that for any $\delta \in
(0,\frac{1}{2}]$, if there exists a polynomial time algorithm that
approximates $\{-\eta^{\frac{1}{2}-\delta},1\}$-MWEB within a
factor of $n^{\epsilon}$, then there exists a polynomial time
algorithm that approximates $\{-1,1\}$-MWEB within a factor of
$n^{(1+ \frac{\delta(3-2\delta)+3}{\delta(1+2\delta)})\epsilon}$
with probability at least $\frac{1}{poly(n)}$. So invoking the
hardness result in Theorem \ref{inapprox1} gives the desired
hardness result for $\{-\eta^{\frac{1}{2}-\delta},1\}$-MWEB.

The same conclusion applies to $\{-1, \eta^{\frac{1}{2} -
\delta}\}$-MWEB by setting $\gamma = 1$, $\alpha = -1$ and $\beta
= (N \eta)^{\frac{1}{2}-\delta}$. Since $\eta$ is a constant for
any given graph, we can simply divide each weight in $\{-1,
\eta^{\frac{1}{2} - \delta}\}$ by $\eta^{\frac{1}{2}-\delta}$.
\qed
\end{proof}

Theorem \ref{BigCoro} leads to the following general statement.

\begin{theorem}\label{generalThm}
For any small constant $\delta \in (0,\frac{1}{2}]$, if $\left|
\frac{\min \mathcal{S} }{\max \mathcal{S}} \right| \in
\Omega(\eta^{\delta-1/2}) \cap O(\eta^{1/2-\delta}) $, then there
exists some constant $\epsilon$ such that no polynomial time
algorithm can approximate $\mathcal{S}$-MWEB within a factor of
$n^\epsilon$ with probability at least $\frac{1}{poly(n)}$ unless
$\mathsf{RP=NP}$.
\end{theorem}

\section{Two Applications}

In this section, we describe two applications of the results
establish in Sect. 3 by proving hardness and inapproximability of
problems found in practice.

\subsection{SAMBA Model is Hard}

Microarray technology has been the latest technological
breakthrough in biological and biomedical research; in many
applications, a key step in analyzing gene expression data
obtained through microarray is the identification of a bicluster
satisfying certain properties and with largest area (see the
survey \cite{Mad04} for a fairly extensive discussion on this).

In particular, Tanay {\it et. al.} \cite{Tan02} considered the
Statistical-Algorithmic Method for Bicluster Analysis (SAMBA)
model. In their formulation, a complete bipartite graph is given
where one side corresponds to genes and the other size corresponds
to conditions. An edges $(u,v)$ is assigned a real weight which
could be either positive or negative, depending on the expression
level of gene $u$ in condition $v$, in a way such that heavy
subgraphs corresponds to statistically significant biclusters. Two
weight-assigning schemes are considered in their paper. In the
first, or simple statistical model, a tight upper-bound on the
probability of an observed biclusters in computed; in the second,
or refined statistical model, the weights are assigned in a way
such that a maximum weight biclique subgraph corresponds to a
maximum likelihood bicluster.

\paragraph{\bf The Simple SAMBA Statistical Model:}  Let $H =
(V_1', V_2', E')$ be a subgraph of $G = (V_1,V_2,E)$,
$\overline{E'} = \{ V_1' \times V_2' \} - E'$ and $p =
\frac{|E|}{|V_1||V_2|}$. The simple statistical model assumes that
edges occur independently and identically at random with
probability $p$. Denote by $BT(k,p,n)$ the probability of
observing $k$ or more successes in $n$ binomial trials, the
probability of observing a graph at least as dense as $H$ is thus
$p(H) = BT(|E'|,p,|V_1'||V_2'|)$. This model assumes $p <
\frac{1}{2}$ and $|V_1'||V_2'| \ll |V_1||V_2|$, therefore $p(H)$
is upper bounded by
$$p^*(H) = 2^{|V_1'||V_2'|} p^{|E'|} (1-p)^{|V_1'||V_2'|-|E'|}$$
The goal of this model is thus to find a subgraph $H$ with the
smallest $p^*(H)$. This is equivalent to maximizing
$$-\log{p^*(H)} = |E'| (-1 - \log{p}) + (|V_1'||V_2'|-|E'|) (-1 -
\log{(1-p)})$$ which is essentially solving a $\mathcal{S}$-MWEB
problem that assigns either positive weight $(-1- \log{p})$ or
negative weight $(-1-\log{(1-p)})$ to an edge $(u,v)$, depending
on whether gene $u$ express or not in condition $v$, respectively.
The summation of edge weights over $H$ is defined as the {\it
statistical significance} of $H$.

Since $\frac{1}{\eta^2} \leq p < \frac{1}{2}$, asymptotically we
have $\frac{-1-\log{(1-p)}}{-1-\log{p}} \in
\Omega(\frac{1}{\log{\eta}}) \cap O(1)$. Invoking Theorem
\ref{generalThm} gives the following.

\begin{theorem} For the Simple SAMBA Statistical
model, there exists some $\epsilon > 0$ such that no polynomial
time algorithm, possibly randomized, can find a bicluster whose
statistical significance is within a factor of $n^{\epsilon}$ of
optimal unless $\mathsf{RP=NP}$.
\end{theorem}

\paragraph{\bf The Refined SAMBA Statistical Model:} In the refined
model, each edge $(u,v)$ is assumed to take an independent
Bernoulli trial with parameter $p_{u,v}$, therefore $p(H) =
(\prod_{(u,v)\in E'} p_{u,v})(\prod_{(u,v)\in \overline{E'}}
(1-p_{u,v}))$ is the probability of observing a subgraph $H$.
Since $p(H)$ generally decreases as the size of $H$ increases,
Tanay {\it et al.} aims to find a bicluster with the largest
(normalized) likelihood ratio $ L(H) = \dfrac{(\prod_{(u,v)\in E'}
p_c)(\prod_{(u,v)\in \overline{E'}} (1-p_c))}{p(H)} $, where $p_c
> \max_{(u,v)\in E} p_{u,v}$ is a constant probability and chosen
with biologically sound assumptions. Note this is equivalent to
maximizing the log-likelihood ratio $$ \log{L(H)} = \displaystyle
\sum_{(u,v)\in E'} \log{\frac{p_c}{p_{u,v}}}  + \sum_{(u,v) \in
\overline{E'} } \log{\frac{1-p_c}{1-p_{u,v}} }  $$ With this
formulation, each edge is assigned weight either $\log{
\frac{p_c}{p_{u,v}} } > 0$ or $\log{ \frac{1-p_c}{1-p_{u,v}} } <
0$ and finding the most statistically significant bicluster is
equivalent to solving $\mathcal{S}$-MWEB with $\mathcal{S} =
\{\log{\frac{1-p_c}{1-p_{u,v}}}, \log{\frac{p_c}{p_{u,v}}} \}$.
Since $p_c$ is a constant and $\frac{1}{\eta^2} \leq p_{u,v} <
p_c$, we have
$\frac{\log{(1-p_c)}-\log{(1-p_{u,v})}}{\log{p_c}-\log{p_{u,v}}}
\in \Omega(\frac{1}{\log{\eta}}) \cap O(1)$. Invoking Theorem
\ref{generalThm} gives the following.

\begin{theorem} For the Refined SAMBA Statistical model,
there exists some $\epsilon > 0$ such that no polynomial time
algorithm, possibly randomized, can find a bicluster whose
log-likelihood is within a factor of $n^{\epsilon}$ of optimal
unless $\mathsf{RP=NP}$.
\end{theorem}

\subsection{Minimum Description Length with Holes (MDLH) is Hard}

Bu {\it et. al} \cite{Bu05} considered the Minimum Description
Length with Holes problem (defined in the following); the
2-dimensional case is claimed $\mathsf{NP}$-hard in this paper and
the proof is referred to \cite{Bu04}. However, the proof given in
\cite{Bu04} suffers from an error in its reduction\footnote{In
Lemma 3.2.1 of \cite{Bu04}, the reduction from CLIQUE to CEW is
incorrect.}, thus whether MDLH is NP-complete remains unsettled.
In this section, by employing the results established in the
previous sections, we show that no polynomial time algorithm
exists for MDLH, under the slightly weaker (than $\mathsf{P \neq
NP}$) but widely believed assumption $\mathsf{RP \neq NP}$.

We first briefly describe the Minimum Description Length
summarization with Holes problem; for a detailed discussion of the
subject, we refer the readers to \cite{Bu04,Bu05}.

Suppose one is given a $k$-dimensional binary matrix $M$, where
each entry is of value either 1, which is of interest, or of value
0, which is not of interest. Besides, there are also $k$
hierarchies (trees) associated with each dimension, namely $T_1,
T_2, ..., T_k$, each of height $l_1, l_2, ..., l_k$ respectively.
Define $level$ $l = \text{max}_i(l_i)$. For each $T_i$, there is a
bijection between its leafs and the 'hyperplanes' in the $i$th
dimension (e.g. in a 2-dimensional matrix, these hyperplanes
corresponds to rows and columns). A $region$ is a tuple
$(x_1,x_2,...,x_k)$, where $x_i$ is a leaf node or an internal
node in hierarchy $T_i$. Region $(x_1,x_2,...,x_k)$ is said to
$cover$ cell $(c_1,c_2,...,c_k)$ if $c_i$ is a descendant of
$x_i$, for all $1 \leq i \leq k$. A {\it $k$-dimensional $l$-level
MDLH summary} is defined as two sets $S$ and $H$, where 1) $S$ is
a set of regions covering all the 1-entries in $M$; and 2) $H$ is
the set of 0-entries covered (undesirably) by $S$ and to be
excluded from the summary. The {\it length} of a summary is
defined as $|S|+|H|$, and the MDLH problem asks the question if
there exists a MDLH summary of length at most $K$, for a given $K
> 0$.

In an effort to establish hardness of MDLH, we first define the
following problem, which serves as an intermediate problem
bridging $\{-1,1\}$-MWEB and MDLH.

\begin{definition}{\bf (Problem $\mathcal{P}$)} \\
\noindent{\bf Instance: } A complete bipartite graph $G =
(V_1,V_2,E)$ where each edge takes on a value in $\{-1,1\}$, and a
positive integer $k$.

\noindent{\bf Question: } Does there exist an induced subgraph (a
biclique $U_1 \times U_2$) whose total weight of edges is
$\omega$, such that $|U_1| + |U_2| + \omega \geq k$.
\end{definition}

\begin{lemma}
\label{problemP} No polynomial time algorithm exists for Problem
$\mathcal{P}$ unless $\mathsf{RP = NP}$.
\end{lemma}
\begin{proof} We prove this by constructing a reduction from
$\{-1,1\}$-MWEB to Problem $\mathcal{P}$ as follows: for the given
input biclique $G = (V_1, V_2, E)$, make $N$ duplicates of $V_1$
and $N$ duplicates of $V_2$, where $N = (|V_1|+|V_2|)^2$. Connect
each copy of $V_1$ to each copy of $V_2$ in a way that is
identical to the input biclique, we then claim that there is a
size $k$ solution to $\{-1,1\}$-MWEB if and only if there is a
size $N^2 k$ solution to Problem $\mathcal{P}$.

If there is a size $k$ solution to $\{-1,1\}$-MWEB, then it is
straightforward that there is a solution to Problem $\mathcal{P}$
of size at least $N^2 k$. For the reverse direction, we show that
if no solution to $\{-1,1\}$-MWEB is of size at least $k$, then
the maximum solution to Problem $\mathcal{P}$ is strictly less
than $N^2 k$. Note a solution $U_1^N \times U_2^N$ to Problem
$\mathcal{P}$ consists of at most $N^2$ (not necessarily all
distinct) solutions to $\{-1,1\}$-MWEB, and each of them can
contribute at most $(k-1)$ in weight to $U_1^N \times U_2^N$, so
the total weight gained from edges is at most $N^2(k-1)$. And note
the total weight gained from vertices is at most $N(|V_1|+|V_2|) =
N\sqrt{N}$, therefore the weight is upper bounded by $N\sqrt{N} +
N^2(k-1) < N^2 k$ and this completes the proof.

As a conclusion, we have a polynomial time reduction from
$\{-1,1\}$-MWEB to Problem $\mathcal{P}$. Since no polynomial time
algorithm exists for $\{-1,1\}$-MWEB unless $\mathsf{RP=NP}$, the
same holds for Problem $\mathcal{P}$. \qed
\end{proof}

\begin{theorem} No polynomial time algorithm exists for
MDLH summarization, even in the 2-dimension 2-level case, unless
$\mathsf{RP=NP}$.
\end{theorem}
\begin{proof} We prove this by showing that Problem $\mathcal{P}$
is a complementary problem of 2-dimensional 2-level MDLH.

Let the input 2D matrix $M$ be of size $n_1 \times n_2$, with a
tree of height 2 associated with each dimension. Without loss of
generality, we only consider the 'sparse' case where the number of
1-entries is less than the number of 0-entries by at least 2 so
that the optimal solution will never contain the whole matrix as
one of its regions. Let $S$ be the set of regions in a solution.
Let $R$ and $C$ be the set of rows and columns not included in
$S$. Let $Z$ be the set of all zero entries in $M$. Let $z$ be the
total number of zero entries in the $R \times C$ 'leftover' matrix
and let $w$ be the total number of 1-entries in it. MDLH tries to
minimize the following:
$$
(n_1-|R|)+(n_2 - |C|) + (|Z|-z) + w = (n_1+n_2+|Z|)-(|R|+|C|+z-w)
$$ Since $(n_1+n_2+|Z|)$ is a fixed quantity for any given input
matrix, the 2-dimensional 2-level MDLH problem is equivalent to
maximizing $(|R|+|C|+z-w)$, which is precisely the definition of
Problem $\mathcal{P}$.

Therefore, 2-dimensional 2-level MDLH is a complementary problem
to Problem $\mathcal{P}$ and by Lemma \ref{problemP} we conclude
that no polynomial time algorithm exists for 2-dimensional 2-level
MDLH unless $\mathsf{RP = NP}$. \qed
\end{proof}

\section{Concluding Remarks}

Maximum weighted edge biclique and its variants have received much
attention in recently years because of it wide range of
applications in various fields including machine learning,
database, and particularly bioinformatics and computational
biology, where many computational problems for the analysis of
microarray data are closely related. To tackle these applied
problems, various kinds of heuristics are proposed and
experimented and it is not known whether these algorithms give
provable approximations. In this work, we answer this question by
showing that it is highly unlikely (under the assumption
$\mathsf{RP \neq NP}$) that good polynomial time approximation
algorithm exists for maximum weighted edge biclique for a wide
range of choices of weight; and we further give specific
applications of this result to two applied problems. We conclude
our work by listing a few open questions.

1. We have shown that $\{\Theta(-\eta^{\delta}), 1\}$-MWEB is
$n^\epsilon$-inapproximable for $\delta \in (-\frac{1}{2},
\frac{1}{2})$; also it is easy to see that (i) the problem is in
$\mathsf{P}$ when $\delta \leq -1$, where the entire input graph
is the optimal solution; (ii) for any $\delta \geq 1$, the problem
is equivalent to MEB, which is conjectured to be
$n^\epsilon$-inapproximable \cite{Fei04}. Therefore it is natural
to ask what is the approximability of the $\{-n^{\delta},
1\}$-MWEB problem when $\delta \in (-1, -\frac{1}{2}]$ and $\delta
\in [\frac{1}{2}, 1]$. In particular, can this be answered by a
better analysis of Lemma \ref{boostingLemma}?

2. We are especially interested in $\{-1,1\}$-MWEB, which is
closely related to the formulations of many natural problems
\cite{Ban04,Bu04,Bu05,Tan02}. We have shown that no polynomial
time algorithm exists for this problem unless $\mathsf{RP = NP}$,
and we believe this problem is $\mathsf{NP}$-complete, however a
proof has eluded us so far.

\end{document}